\documentclass[letter,twocolumn]{jpsj3}

\usepackage{amsmath,amssymb,amsthm,bm}
\usepackage{txfonts}
\usepackage{graphicx}

\newtheorem{theorem}{Theorem}
\newtheorem{corollary}[theorem]{Corollary}
\newtheorem{proposition}[theorem]{Proposition}

\newcommand{\ii}{\mathrm{i}}
\newcommand{\ee}{\mathrm{e}}
\newcommand{\E}{\mathbb{E}}

\newcommand{\supp}{\operatorname{supp}}
\newcommand{\dist}{\operatorname{dist}}
\newcommand{\sinc}{\operatorname{sinc}}
\newcommand{\SPSA}{\mathcal{S}_{c}}

\title{Simultaneous Perturbation as a Spectral Filter for Parametrized Quantum Circuits}
\author{Masayuki \textsc{Ohzeki}$^{1,2,3,4}$\thanks{E-mail address: mohzeki@tohoku.ac.jp}}
\inst{%
$^1$Graduate School of Information Sciences, Tohoku University, Sendai 980-8579, Japan\\
$^2$Department of Physics, Institute of Science Tokyo, Tokyo 152-8551, Japan\\
$^3$Research and Education Institute for Semiconductors and Informatics, Kumamoto University, Kumamoto 860-8555, Japan\\
$^4$Sigma-i Co., Ltd., Tokyo 108-0075, Japan
}
\recdate{\today}

\abst{Parameterized quantum circuits require repeated measurements to estimate objective gradients, making the number of circuit evaluations a central resource. 
Simultaneous perturbation stochastic approximation (SPSA) is widely used because it obtains a full gradient estimate from only two cost evaluations, independently of the number of parameters. 
We analyze the random perturbation for the structured Fourier landscapes generated by quantum circuits. 
For a finite update width, we derive the exact mode-by-mode response of the mean Rademacher estimator and show that it selectively suppresses modes involving many parameters.
For general generators, the filtered mean field need not be conservative. 
We derive its sampling variance, compare it with Gaussian random-direction smoothing, and formulate the resulting stochastic dynamics through a Fokker--Planck equation. 
The analysis shows that SPSA is a geometry-dependent spectral filter: decreasing $c$ provides a continuation from a smoothed landscape toward the original objective, while the associated diffusion controls exploration. }

\begin{document}

\maketitle

\noindent{\it Introduction.---}
Parametrized quantum circuits (PQCs) encode variational states and predictions in expectation values controlled by classical parameters.
This construction underlies the variational quantum eigensolver, the quantum approximate optimization algorithm, and quantum circuit learning.
On present devices, the objective is obtained from repeated measurements and is therefore both expensive and noisy.
Estimating one objective value already requires many shots, and a conventional gradient step repeats this measurement cost for each parameter direction. Gradient evaluation can therefore dominate the runtime and sampling budget, even for shallow circuits \cite{Kandala2017,Lavrijsen2020,BonetMonroig2023}.

Simultaneous perturbation stochastic approximation (SPSA) avoids this scaling. It chooses a random direction $\bm\Delta$ with independent entries $\Delta_j=\pm1$ of equal probability; this symmetric Bernoulli law is the Rademacher distribution. For a width $c>0$, the two evaluations $F(\bm\theta\pm c\bm\Delta)$ yield a full gradient estimate, so each update uses only two circuit evaluations, independent of parameter number \cite{Spall1992,Spall1998}. This economy has motivated applications in quantum tomography, variational optimization, and quantum geometric tensors \cite{Ferrie2014,Gacon2021,Periyasamy2024}. The randomness is usually discussed statistically; here we ask what transformation it induces on the circuit objective.

Quantum-circuit objectives possess additional structure not resolved by a local Taylor estimate of the SPSA bias.
Expectation values generated by unitary parameter gates form finite multivariate Fourier series, whose frequencies are fixed by differences of generator eigenvalues \cite{Schuld2021,Wierichs2022,Nemkov2023}.
This representation has been used to characterize expressivity, derive parameter-shift rules, recover cost landscapes, and relate spectral power to barren plateaus \cite{Stechly2023,OkumuraOhzeki2026}.
It is therefore natural to ask how a stochastic gradient estimator acts on the Fourier components themselves.

This also gives a dynamical interpretation. The expected update smooths the Fourier landscape, while finite-sample fluctuations act as diffusion and lead, in a diffusion approximation, to a Fokker--Planck equation. In the reversible case its stationary density is set by the smoothed objective and an effective temperature determined by the learning rate and estimator variance. Decreasing the perturbation width and this temperature is therefore an annealing protocol: sufficiently slow cooling can concentrate on global minima, whereas practical finite schedules may lose diffusion and become trapped at local minima, although standard stochastic-gradient results establish only stationary-point convergence. We derive the exact spectral action, sampling variance, and Fokker--Planck description, clarifying when SPSA regularizes a rugged landscape and when it removes useful descent information.

\noindent{\it Fourier structure.---}
Consider a depth-$d$ circuit containing $L$ independently parametrized gates, with one parameter per gate for notational simplicity,
\begin{equation}
U(\bm\theta)=W_{L+1}\prod_{j=L}^{1}\ee^{-\ii\theta_jG_j}W_j,
\label{eq:circuit}
\end{equation}
where the $W_j$ are fixed unitaries and the $G_j$ are finite-dimensional Hermitian generators.
For an input state $\varrho$ and an observable $O$, the objective is
\begin{equation}
F(\bm\theta)=\operatorname{Tr}\!\left[O U(\bm\theta)\varrho U(\bm\theta)^\dagger\right].
\label{eq:objective}
\end{equation}
Let $\Lambda_j$ be the spectrum of $G_j$ and let
\begin{equation}
\Omega_j=\{\lambda-\lambda':\lambda,\lambda'\in\Lambda_j\}
\label{eq:frequency-set}
\end{equation}
be its difference spectrum.
Let $\Pi_{j,\lambda}$ denote the spectral projector of $G_j$ onto the eigenspace with eigenvalue $\lambda$. Inserting the spectral resolution $\ee^{-\ii\theta_jG_j}=\sum_{\lambda\in\Lambda_j}\ee^{-\ii\theta_j\lambda}\Pi_{j,\lambda}$ on both sides of $\varrho$ gives
\begin{equation}
F(\bm\theta)=
\sum_{\bm\omega\in\Omega_1\times\cdots\times\Omega_L}
\widehat F_{\bm\omega}\ee^{\ii\bm\omega\cdot\bm\theta}.
\label{eq:fourier-expansion}
\end{equation}
Hermiticity implies $\widehat F_{-\bm\omega}=\widehat F_{\bm\omega}^*$.
Equation~\eqref{eq:fourier-expansion} is a finite generalized Fourier series; for commensurate spectra it is an ordinary Fourier series on a torus \cite{Schuld2021,Wierichs2022}.

For a Pauli rotation $G_j=P_j/2$, one has $\Omega_j=\{-1,0,1\}$.
Although every coordinate therefore has degree at most one, the support
\begin{equation}
|\supp\bm\omega|=|\{j:\omega_j\ne0\}|
\label{eq:support}
\end{equation}
can grow with circuit depth and parameter count.
We refer to $|\supp\bm\omega|$ as the Fourier order or support size.

\noindent{\it Exact SPSA multiplier.---}
Let $\bm\Delta=(\Delta_1,\ldots,\Delta_L)$ have independent Rademacher entries, meaning $\Delta_j\in\{+1,-1\}$ with $\Pr(\Delta_j=1)=\Pr(\Delta_j=-1)=1/2$. Let $c>0$ denote the perturbation width.
The two-sided SPSA estimator is
\begin{equation}
\widehat g_i(\bm\theta;\bm\Delta)
=\frac{F(\bm\theta+c\bm\Delta)-F(\bm\theta-c\bm\Delta)}
{2c\Delta_i}.
\label{eq:spsa-estimator}
\end{equation}
Because $\Delta_i^{-1}=\Delta_i$, averaging over independently sampled directions defines the mean, or population, operator
\begin{equation}
(\SPSA F)_i(\bm\theta)=\E_{\bm\Delta}\!\left[\widehat g_i(\bm\theta;\bm\Delta)\right].
\label{eq:population-operator}
\end{equation}

\begin{theorem}[Rademacher-SPSA Fourier multiplier]
For the Fourier mode $e_{\bm\omega}(\bm\theta)=
\ee^{\ii\bm\omega\cdot\bm\theta}$,
\begin{equation}
(\SPSA e_{\bm\omega})_i
=\ii\omega_i\rho_i(\bm\omega;c)e_{\bm\omega},
\label{eq:multiplier-action}
\end{equation}
where
\begin{equation}
\rho_i(\bm\omega;c)=
\sinc(c\omega_i)\prod_{j\ne i}\cos(c\omega_j)
\label{eq:rho}
\end{equation}
and $\sinc x=\sin x/x$, with $\sinc0=1$.
\end{theorem}

\begin{proof}
Substitution of $e_{\bm\omega}$ into Eq.~\eqref{eq:spsa-estimator} gives
\begin{equation}
\widehat g_i=
\frac{\ii}{c}e_{\bm\omega}(\bm\theta)
\Delta_i\sin(c\bm\omega\cdot\bm\Delta).
\end{equation}
Independence of the Rademacher variables yields
\begin{align}
\E_{\bm\Delta}\!\left[
\Delta_i\ee^{\ii c\bm\omega\cdot\bm\Delta}
\right]
&=\E_{\Delta_i}[\Delta_i\ee^{\ii c\omega_i\Delta_i}]
\prod_{j\ne i}\E_{\Delta_j}[\ee^{\ii c\omega_j\Delta_j}]\nonumber\\
&=\ii\sin(c\omega_i)\prod_{j\ne i}\cos(c\omega_j).
\end{align}
Taking the imaginary part and multiplying by $\ii e_{\bm\omega}/c$ proves Eqs.~\eqref{eq:multiplier-action} and \eqref{eq:rho}.
\end{proof}

The theorem is exact at finite $c$ and applies mode by mode to Eq.~\eqref{eq:fourier-expansion}.
It refines the usual statement $\E\widehat{\bm g}=\nabla F+O(c^2)$ by retaining the full spectrum of the bias.
In general this vector field is not conservative. For a mode with $\omega_i\omega_j\ne0$, its contributions to $\partial_j(\SPSA F)_i$ and $\partial_i(\SPSA F)_j$ are proportional to $\rho_i$ and $\rho_j$, respectively. They therefore agree only when the two coordinate-dependent multipliers happen to be equal. The random update consequently need not be the gradient of one common smoothed objective; its mean can instead define a genuinely nonconservative drift.

\noindent{\it Suppression by Fourier support.---}
To state the attenuation bound, first measure the distance of a phase from the multiples of $\pi$ by
\begin{equation}
d_\pi(x)=\dist(x,\pi\mathbb Z)=\min_{n\in\mathbb Z}|x-n\pi|.
\end{equation}
A nonzero coordinate of $\bm\omega$ does not by itself imply attenuation: $|\cos x|$ returns to one at $x\in\pi\mathbb Z$. Thus $d_\pi(c\omega_j)$ identifies whether the $j$th factor genuinely suppresses the mode.

\begin{corollary}[Nonresonant support bound]
Let $0<\delta\le\pi/2$ and suppose that a set
$S_i\subseteq\{1,\ldots,L\}\setminus\{i\}$ satisfies
\begin{equation}
d_\pi(c\omega_j)\ge\delta,
\qquad j\in S_i.
\end{equation}
Then
\begin{equation}
|\rho_i(\bm\omega;c)|
\le |\sinc(c\omega_i)|(\cos\delta)^{|S_i|}.
\label{eq:nonresonant-bound}
\end{equation}
\end{corollary}

Since $d_\pi(x)\in[0,\pi/2]$ and $|\cos x|=\cos d_\pi(x)$, each selected factor is at most $\cos\delta$. Equation~\eqref{eq:nonresonant-bound} therefore decreases as $(\cos\delta)^{|S_i|}$: every additional nonresonant parameter multiplies the response by a number smaller than one, producing exponential suppression in their number. In applying the bound, one should check $d_\pi(c\omega_j)$ rather than only $|c\omega_j|$; phases close to a multiple of $\pi$ are resonant and are not suppressed.

The most relevant quantum-circuit case is simpler because the allowed frequencies are fixed.

\begin{corollary}[Pauli rotations]
Suppose $\omega_j\in\{-1,0,1\}$ and $0<c<\pi/2$.
For a mode of support $m$ with $\omega_i\ne0$,
\begin{equation}
|\rho_i(\bm\omega;c)|
=\sinc(c)(\cos c)^{m-1}.
\label{eq:pauli-filter}
\end{equation}
Hence every fixed nonzero $c$ produces exponential attenuation in the Fourier order $m$.
\end{corollary}

For this Pauli spectrum, $\rho_i$ is the same for every active coordinate of a mode. Thus the population drift is the gradient of a single filtered objective, $\SPSA F=\nabla F_c,$ and
\begin{equation}
F_c=\widehat F_{\bm0}+\sum_{\bm\omega\ne\bm0}
r_{|\supp\bm\omega|}\widehat F_{\bm\omega}\ee^{\ii\bm\omega\cdot\bm\theta},
\label{eq:pauli-potential}
\end{equation}
where $r_m=\sinc(c)(\cos c)^{m-1}$. Thus, for independently parametrized Pauli rotations, SPSA is exactly gradient descent on an explicitly filtered landscape; this special potential structure is absent for general generator spectra.

Equation~\eqref{eq:pauli-filter} distinguishes coordinate frequency from interaction order.
A mode such as $\cos(q\theta_1)$ oscillates rapidly in one coordinate but has support one, whereas $\cos(\theta_1+\cdots+\theta_m)$ has unit frequency in each coordinate but support $m$. Rademacher SPSA therefore filters these modes differently: the former is not attenuated by the transverse product, while the latter is exponentially suppressed by Eq.~\eqref{eq:pauli-filter}. The distinction is the number of coupled parameters, because it determines how many cosine factors act on the mode.

For small $c$, Eq.~\eqref{eq:rho} gives a complementary norm-dependent form,
\begin{equation}
\log|\rho_i|
=-\frac{c^2\omega_i^2}{6}
-\frac{c^2}{2}\sum_{j\ne i}\omega_j^2
+O\!\left(c^4\sum_j\omega_j^4\right).
\label{eq:small-c}
\end{equation}
The response changes by an order-one amount when $c^2\sum_{j\ne i}\omega_j^2=O(1)$. For comparable unit frequencies, this condition becomes
$c^2(m-1)=O(1)$, or $c=O(m^{-1/2})$.
Thus, at fixed $c$, the estimator acts as a support-selective
filter: low-support modes remain visible, whereas modes involving
many parameters are exponentially attenuated. This is the mechanism
by which a finite-width SPSA update can smooth a rugged circuit
landscape, although useful high-support information may be removed
at the same time.

There is also a typical-case statement that does not require one common $\delta$.
If the $m-1$ active transverse phases $X_j=c\omega_j\bmod\pi$ are independent and identically distributed with
$0<-\E\log|\cos X_j|=\kappa<\infty$, then the strong law of large numbers gives
$\log\left|
\prod_{j\ne i}\cos X_j
\right|/(m-1) \to -\kappa$
almost surely along increasing support; the product is over these active coordinates.
For phases uniform modulo $\pi$, $\kappa=\log2$.
The logarithmic decay rate is then $\log2$. This conclusion applies to a frequency ensemble whose phases have the stated distribution; it is not a claim that every fixed circuit has independent uniform phases.

\noindent{\it Finite-direction variance.---}
The filter in Eq.~\eqref{eq:rho} describes the expectation over perturbation directions.
A single SPSA estimate is not itself a smoothed gradient and may have large variance.
This distinction is essential for interpreting the mechanism.

\begin{proposition}[Single-mode variance]
For the complex unit-amplitude mode $e_{\bm\omega}$, with $\operatorname{Var}X=\E|X-\E X|^2$,
\begin{align}
\E|\widehat g_i|^2
&=\frac{1-\prod_j\cos(2c\omega_j)}{2c^2},
\label{eq:second-moment}\\
\operatorname{Var}(\widehat g_i)
&=\frac{1-\prod_j\cos(2c\omega_j)}{2c^2}
-\frac{\sin^2(c\omega_i)}{c^2}
\prod_{j\ne i}\cos^2(c\omega_j).
\label{eq:variance}
\end{align}
\end{proposition}

\begin{proof}
For a Fourier mode,
$|\widehat g_i|^2=c^{-2}\sin^2(c\bm\omega\cdot\bm\Delta)$.
Using $\sin^2x=(1-\cos2x)/2$ and Rademacher independence gives Eq.~\eqref{eq:second-moment}.
Subtracting the squared modulus of Eq.~\eqref{eq:multiplier-action} gives Eq.~\eqref{eq:variance}.
\end{proof}

\noindent{\it Gaussian random directions.---}
The product response is specific to the hypercube vertices sampled by Rademacher perturbations.
For comparison, let $\bm Z\sim\mathcal N(0,I)$ and consider the two-sided random-direction estimator
\begin{equation}
\widehat h_i=
\frac{F(\bm\theta+c\bm Z)-F(\bm\theta-c\bm Z)}{2c}Z_i.
\end{equation}
For a Fourier mode, the needed average follows by differentiating the Gaussian characteristic function,
\begin{equation}
\E e^{\ii c\bm\omega\cdot\bm Z}=e^{-c^2\|\bm\omega\|_2^2/2},    
\end{equation}
yields
\begin{equation}
\E\widehat h_i
=\sum_{\bm\omega}
\ii\omega_i\widehat F_{\bm\omega}
\ee^{\ii\bm\omega\cdot\bm\theta}
\ee^{-c^2\|\bm\omega\|_2^2/2}.
\label{eq:gaussian-filter}
\end{equation}
This is the gradient of the Gaussian-smoothed objective and is a radial low-pass filter \cite{NesterovSpokoiny2017,Gao2022}.

The two perturbation laws therefore retain different spectral information.
Gaussian smoothing treats all frequency vectors with the same Euclidean norm identically.
Rademacher SPSA distinguishes how that norm is distributed among coordinates and has exact zeros and revivals associated with the hypercube geometry.
For PQCs with Pauli rotations, the support and squared norm coincide, $\|\bm\omega\|_2^2=m$, and both responses decay exponentially in $m$ for small $c$.
Beyond that special case, the perturbation law becomes a genuine filter-design choice.  Isotropic Gaussian directions impose a radial cutoff determined only by $\|\bm\omega\|_2$, whereas Rademacher directions retain the coordinate structure of $\bm\omega$ and preferentially attenuate modes whose nonresonant frequency content is distributed over many parameters.  The latter is advantageous when useful descent information is concentrated within a few parameter groups while broadly supported modes mainly produce landscape roughness; Gaussian smoothing is more appropriate when the relevance of a mode is governed by its total spectral norm.  This choice can be refined anisotropically: replacing $c\Delta_j$ by $c_j\Delta_j$ changes the Rademacher response to $\sinc(c_i\omega_i)\prod_{j\ne i}\cos(c_j\omega_j)$, while a preconditioned Gaussian estimator with covariance $C$ produces the factor $\exp[-c^2\bm\omega^{\mathsf T}C\bm\omega/2]$.  The widths $c_j$ or the metric $C$ can therefore preserve an anticipated signal sector while suppressing nuisance sectors.

\noindent{\it Stochastic dynamics.---}
Let $\mathcal F_k$ denote the information available before iteration $k$. The fresh direction and measurement outcomes may be assumed conditionally independent of the past given $\bm\theta_k$. 
Writing $\bm m_c(\bm\theta_k)=\E[\widehat{\bm g}_k\mid\mathcal F_k]$ and $\bm\xi_k=\widehat{\bm g}_k-\bm m_{c_k}(\bm\theta_k)$ gives $\E[\bm\xi_k\mid\mathcal F_k]=0$ and $\E[\bm\xi_k\bm\xi_k^{\mathsf T}\mid\mathcal F_k]=\Sigma_c(\bm\theta_k)$. Thus $\bm\xi_k$ is the combined martingale fluctuation from the random direction and the measured costs.
Substitution into the usual descending update
$\bm\theta_{k+1}=\bm\theta_k-a_k\widehat{\bm g}_k$ gives
\begin{equation}
\bm\theta_{k+1}=\bm\theta_k-a_k\bigl[\bm m_{c_k}(\bm\theta_k)+\bm\xi_k\bigr],
\label{eq:spsa-update}
\end{equation}
Take $a_k=a\ll1$ and $c_k=c$ over the time window considered. At the grid time $t_N=Na$, summing Eq.~\eqref{eq:spsa-update} gives the exact decomposition
\begin{equation}
\bm\theta_N=\bm\theta_0-a\sum_{k=0}^{N-1}\bm m_c(\bm\theta_k)
-\sqrt a\,\bm M^a(t_N),
\label{eq:accumulated-update}
\end{equation}
where $\bm M^a(t)=\sqrt a\sum_{k<t/a}\bm\xi_k$.
Thus the term that becomes diffusion is already present in the discrete dynamics as the accumulated fluctuation $\sqrt a\,\bm M^a$. Assume that $\bm m_c$ and $\Sigma_c$ are locally Lipschitz, the iterates remain in a bounded region, and $\bm\xi_k$ has a uniform $(2+\epsilon)$ moment. Then the predictable quadratic variation is $\langle\bm M^a\rangle_t=a\sum_{k<t/a}\Sigma_c(\bm\theta_k)$, which converges to $\int_0^t\Sigma_c(\bm\theta(s))\,\mathrm ds$. The martingale functional central-limit theorem yields $\bm M^a\Rightarrow\int_0^t\Sigma_c^{1/2}\,\mathrm d\bm W_s$ \cite{KushnerYin2003}, while the deterministic sum in Eq.~\eqref{eq:accumulated-update} becomes $\int_0^t\bm m_c(\bm\theta(s))\,\mathrm ds$. Since $-\bm W_t$ has the same law as $\bm W_t$, the resulting weak approximation is
\begin{equation}
\mathrm d\bm\theta=-\bm m_c(\bm\theta)\,\mathrm dt
+\sqrt a\,\Sigma_c(\bm\theta)^{1/2}\,\mathrm d\bm W_t.
\label{eq:spsa-sde}
\end{equation}
The strict $a\to0$ limit is $\dot{\bm\theta}=-\bm m_c$; Eq.~\eqref{eq:spsa-sde} retains the leading $O(\sqrt a)$ fluctuation of the discrete recursion. Its density obeys
\begin{equation}
\partial_t p=\nabla\!\cdot(\bm m_c p)+\frac{a}{2}\partial_i\partial_j\bigl(\Sigma_{c,ij}p\bigr),
\label{eq:spsa-fp}
\end{equation}
The first term is the mean-update drift and the second is diffusion from direction and measurement fluctuations. If the mean field is conservative, define $F_c$ by $\bm m_c=\nabla F_c$; for Pauli rotations this is Eq.~\eqref{eq:pauli-potential}. For constant isotropic noise $\Sigma_c=\sigma_c^2I$, the reversible stationary solution is
$p_{\rm st}(\bm\theta)\propto\exp[-2F_c(\bm\theta)/(a\sigma_c^2)]$.
Thus $T_{\rm eff}=a\sigma_c^2/2$, while $c$ selects the smoothed potential. As $c_k\to0$, $F_c\to F$; sufficiently slow logarithmic cooling of $T_{\rm eff}$ gives concentration on global minima under the Geman--Geman assumptions \cite{GemanGeman1984}. Finite cooling can instead leave the process near a local minimum. For a nonconservative field, the stationary law follows the full drift--diffusion balance and generally carries probability currents. Stationary-density-preserving nonreversible currents have likewise been used to violate detailed balance and accelerate relaxation without changing the target distribution \cite{IchikiOhzeki2013,OhzekiIchiki2015,Ohzeki2026Gauge}. 
For one Fourier mode, this rotation is explicit. With
$m_i=\ii\omega_i\rho_i\widehat F_{\bm\omega}\ee^{\ii\bm\omega\cdot\bm\theta}$, direct differentiation gives
\begin{equation}
\partial_jm_i-\partial_im_j
=-\omega_i\omega_j\widehat F_{\bm\omega}e^{\ii\bm\omega\cdot\bm\theta}(\rho_i-\rho_j).
\label{eq:spsa-curl}
\end{equation}
Indeed, $\partial_jm_i=-\omega_i\omega_j\rho_i\widehat F_{\bm\omega}\ee^{\ii\bm\omega\cdot\bm\theta}$ and $\partial_im_j$ is the same expression with $\rho_i$ replaced by $\rho_j$, which proves Eq.~\eqref{eq:spsa-curl}. Expanding Eq.~\eqref{eq:rho} for small $c$ gives $\rho_i-\rho_j=c^2(\omega_i^2-\omega_j^2)/3+O(c^4)$.
Thus the rotation is $O(c^2)$ for general generator spectra. For Pauli rotations, every active coordinate has $\omega_i^2=\omega_j^2=1$, so $\rho_i=\rho_j$ exactly and Eq.~\eqref{eq:spsa-curl} vanishes mode by mode.

\noindent{\it Implications and limitations.---}
The standard SPSA schedule decreases the perturbation width, often as
$c_k=c_0/(k+1)^\gamma$ \cite{Spall1992,Spall1998}.
Equations~\eqref{eq:rho} and \eqref{eq:small-c} give this prescription a spectral-continuation interpretation.
At early iterations, finite $c_k$ suppresses high-support components of the population drift.
As $c_k$ decreases, progressively finer and more strongly coupled modes reenter, and $\rho_i\to1$ for every fixed $\bm\omega$.
Thus $c_k$ controls the spectral bias of the conditional mean, whereas $a_k$ controls the magnitude of both the deterministic update and the accumulated fluctuation in Eq.~\eqref{eq:spsa-sde}.  The exact result is that every Fourier contribution to the mean update is multiplied by $\rho_i(\bm\omega;c_k)$.  A finite width is beneficial only when the strongly attenuated modes are less useful for descent than the retained modes; otherwise it removes part of the optimization signal.

The mechanism must not be confused with a cure for barren plateaus.
A barren plateau suppresses the spectral power of the objective or its true gradient itself \cite{McClean2018,Cerezo2021,OkumuraOhzeki2026}.
SPSA retains its practical advantage that one stochastic gradient estimate uses only two cost evaluations, independently of parameter number.  However, all components are inferred from the same random-direction cost difference.  Equation~\eqref{eq:variance} shows that attenuation of its mean does not entail a comparable reduction of the single-direction variance, and measurement noise adds further fluctuations to that difference.  Consequently, if the true gradient is exponentially small, SPSA does not restore the missing signal: maintaining a fixed relative accuracy requires additional directions or shots in proportion to the variance divided by the squared mean.  Its spectral filtering can suppress selected contributions to the mean landscape, but it cannot create gradient spectral weight that is absent in a barren plateau.

\noindent{\it Conclusion.---}
We have derived the exact finite-width response of the Rademacher-SPSA gradient estimator on a multivariate Fourier basis. It is a product filter that suppresses nonresonant modes extending over many parameters; for Pauli-rotation circuits it becomes $\sinc(c)(\cos c)^{m-1}$. The effect is distinct from radial Gaussian smoothing and is useful only when the attenuated modes are predominantly rough rather than informative. The same multiplier also determines the finite-width drift and its possible nonconservative component in Eq.~\eqref{eq:spsa-curl}.

\begin{acknowledgments}
This work was supported by the Cabinet Office, Government of Japan, through the Cross-ministerial Strategic Innovation Promotion Program (SIP), and by the Ministry of Education, Culture, Sports, Science and Technology (MEXT) through the Quantum Leap Flagship Program (MEXT Q-LEAP).
\end{acknowledgments}

\bibliographystyle{jpsj}
\bibliography{references}

@article{Spall1992,
  author = {Spall, James C.},
  title = {Multivariate Stochastic Approximation Using a Simultaneous Perturbation Gradient Approximation},
  journal = {IEEE Transactions on Automatic Control},
  volume = {37},
  number = {3},
  pages = {332--341},
  year = {1992},
  doi = {10.1109/9.119632}
}

@article{Spall1998,
  author = {Spall, James C.},
  title = {An Overview of the Simultaneous Perturbation Method for Efficient Optimization},
  journal = {Johns Hopkins APL Technical Digest},
  volume = {19},
  number = {4},
  pages = {482--492},
  year = {1998},
  url = {https://www.jhuapl.edu/spsa/pdf-spsa/spall_an_overview.pdf}
}

@article{NesterovSpokoiny2017,
  author = {Nesterov, Yurii and Spokoiny, Vladimir},
  title = {Random Gradient-Free Minimization of Convex Functions},
  journal = {Foundations of Computational Mathematics},
  volume = {17},
  pages = {527--566},
  year = {2017},
  doi = {10.1007/s10208-015-9296-2}
}

@inproceedings{Gao2022,
  author = {Gao, Katelyn and Sener, Ozan},
  title = {Generalizing Gaussian Smoothing for Random Search},
  booktitle = {Proceedings of the 39th International Conference on Machine Learning},
  series = {Proceedings of Machine Learning Research},
  volume = {162},
  pages = {7077--7101},
  year = {2022},
  url = {https://proceedings.mlr.press/v162/gao22f.html}
}

@article{Schuld2021,
  author = {Schuld, Maria and Sweke, Ryan and Meyer, Johannes Jakob},
  title = {Effect of Data Encoding on the Expressive Power of Variational Quantum-Machine-Learning Models},
  journal = {Physical Review A},
  volume = {103},
  pages = {032430},
  year = {2021},
  doi = {10.1103/PhysRevA.103.032430}
}

@article{Wierichs2022,
  author = {Wierichs, David and Izaac, Josh and Wang, Cody and Lin, Cedric Yen-Yu},
  title = {General Parameter-Shift Rules for Quantum Gradients},
  journal = {Quantum},
  volume = {6},
  pages = {677},
  year = {2022},
  doi = {10.22331/q-2022-03-30-677}
}

@article{Nemkov2023,
  author = {Nemkov, Nikita A. and Kiktenko, Evgeniy O. and Fedorov, Aleksey K.},
  title = {Fourier Expansion in Variational Quantum Algorithms},
  journal = {Physical Review A},
  volume = {108},
  pages = {032406},
  year = {2023},
  doi = {10.1103/PhysRevA.108.032406}
}

@article{Stechly2023,
  author = {Stechly, Michal and Gao, Lanruo and Yogendran, Boniface and Fontana, Enrico and Rudolph, Manuel},
  title = {Connecting the Hamiltonian Structure to the {QAOA} Energy and Fourier Landscape Structure},
  journal = {arXiv:2305.13594},
  year = {2023},
  doi = {10.48550/arXiv.2305.13594}
}

@article{OkumuraOhzeki2026,
  author = {Okumura, Shun and Ohzeki, Masayuki},
  title = {Fourier Analysis of Parameterized Quantum Circuits and the Barren Plateau Problem},
  journal = {Journal of the Physical Society of Japan},
  volume = {95},
  pages = {033002},
  year = {2026},
  doi = {10.7566/JPSJ.95.033002}
}

@article{Ferrie2014,
  author = {Ferrie, Christopher},
  title = {Self-Guided Quantum Tomography},
  journal = {Physical Review Letters},
  volume = {113},
  pages = {190404},
  year = {2014},
  doi = {10.1103/PhysRevLett.113.190404}
}

@article{Kandala2017,
  author = {Kandala, Abhinav and Mezzacapo, Antonio and Temme, Kristan and Takita, Maika and Brink, Markus and Chow, Jerry M. and Gambetta, Jay M.},
  title = {Hardware-Efficient Variational Quantum Eigensolver for Small Molecules and Quantum Magnets},
  journal = {Nature},
  volume = {549},
  pages = {242--246},
  year = {2017},
  doi = {10.1038/nature23879}
}

@inproceedings{Lavrijsen2020,
  author = {Lavrijsen, Wim and Tudor, Ana and M{\"u}ller, Juliane and Iancu, Costin and de Jong, Wibe},
  title = {Classical Optimizers for Noisy Intermediate-Scale Quantum Devices},
  booktitle = {2020 IEEE International Conference on Quantum Computing and Engineering (QCE)},
  pages = {267--277},
  year = {2020},
  doi = {10.1109/QCE49297.2020.00041}
}

@article{Gacon2021,
  author = {Gacon, Julien and Zoufal, Christa and Carleo, Giuseppe and Woerner, Stefan},
  title = {Simultaneous Perturbation Stochastic Approximation of the Quantum Fisher Information},
  journal = {Quantum},
  volume = {5},
  pages = {567},
  year = {2021},
  doi = {10.22331/q-2021-10-20-567}
}

@article{BonetMonroig2023,
  author = {Bonet-Monroig, Xavier and Wang, Hao and Vermetten, Diederick and Senjean, Bruno and Moussa, Charles and B{\"a}ck, Thomas and Dunjko, Vedran and O'Brien, Thomas E.},
  title = {Performance Comparison of Optimization Methods on Variational Quantum Algorithms},
  journal = {Physical Review A},
  volume = {107},
  pages = {032407},
  year = {2023},
  doi = {10.1103/PhysRevA.107.032407}
}

@inproceedings{Periyasamy2024,
  author = {Periyasamy, Maniraman and Plinge, Axel and Mutschler, Christopher and Scherer, Daniel D. and Mauerer, Wolfgang},
  title = {Guided-{SPSA}: Simultaneous Perturbation Stochastic Approximation Assisted by the Parameter Shift Rule},
  booktitle = {2024 IEEE International Conference on Quantum Computing and Engineering (QCE)},
  pages = {1504--1515},
  year = {2024},
  doi = {10.1109/QCE60285.2024.00177}
}

@article{McClean2018,
  author = {McClean, Jarrod R. and Boixo, Sergio and Smelyanskiy, Vadim N. and Babbush, Ryan and Neven, Hartmut},
  title = {Barren Plateaus in Quantum Neural Network Training Landscapes},
  journal = {Nature Communications},
  volume = {9},
  pages = {4812},
  year = {2018},
  doi = {10.1038/s41467-018-07090-4}
}

@article{Cerezo2021,
  author = {Cerezo, M. and Sone, Akira and Volkoff, Tyler and Cincio, Lukasz and Coles, Patrick J.},
  title = {Cost Function Dependent Barren Plateaus in Shallow Parametrized Quantum Circuits},
  journal = {Nature Communications},
  volume = {12},
  pages = {1791},
  year = {2021},
  doi = {10.1038/s41467-021-21728-w}
}

@article{IchikiOhzeki2013,
  author = {Ichiki, Akihisa and Ohzeki, Masayuki},
  title = {Violation of detailed balance accelerates relaxation},
  journal = {Physical Review E},
  volume = {88},
  pages = {020101(R)},
  year = {2013},
  doi = {10.1103/PhysRevE.88.020101}
}

@article{OhzekiIchiki2015,
  author = {Ohzeki, Masayuki and Ichiki, Akihisa},
  title = {Langevin dynamics neglecting detailed balance condition},
  journal = {Physical Review E},
  volume = {92},
  pages = {012105},
  year = {2015},
  doi = {10.1103/PhysRevE.92.012105}
}

@article{Ohzeki2026Gauge,
  author = {Ohzeki, Masayuki},
  title = {Nonreversible Gauge Fields in {Fokker--Planck} Dynamics: Supersymmetric Hamiltonians and Learned Finite Forces},
  journal = {Physical Review Research},
  volume = {8},
  pages = {033170},
  year = {2026},
  doi = {10.1103/tg4c-7cnr}
}

@article{GemanGeman1984,
  author = {Geman, Stuart and Geman, Donald},
  title = {Stochastic relaxation, Gibbs distributions, and the Bayesian restoration of images},
  journal = {IEEE Transactions on Pattern Analysis and Machine Intelligence},
  volume = {6},
  pages = {721--741},
  year = {1984},
  doi = {10.1109/TPAMI.1984.4767596}
}

@book{KushnerYin2003,
  author = {Kushner, Harold J. and Yin, G. George},
  title = {Stochastic Approximation and Recursive Algorithms and Applications},
  edition = {2},
  publisher = {Springer},
  address = {New York},
  year = {2003},
  doi = {10.1007/b97441}
}

\end{document}